\documentclass[runningheads]{llncs}
\usepackage{graphicx} % Required for inserting images
\usepackage{todonotes}
\usepackage{color}
\usepackage{amsmath, amsfonts, amssymb}
\usepackage[shortlabels,inline]{enumitem}

\let\doendproof\endproof
\renewcommand\endproof{~\hfill$\qed$\doendproof}

\setlist[enumerate]{leftmargin=1.2cm}

\newtheorem{observation}[lemma]{Observation}

\title{Approximation Ratio of the Min-Degree Greedy Algorithm for Maximum Independent Set on Interval and Chordal Graphs}

\author{Steven Chaplick\inst{1}, Martin Frohn\inst{1}, Steven Kelk\inst{1},\\ Johann Lottermoser\inst{1}, Mat\'u\v{s} Mihal\'ak\inst{1}}
\institute{Department of Advanced Computing Sciences (DACS),\\ Maastricht University, The Netherlands.}
\date{December 2023}

\begin{document}

\maketitle

\begin{abstract}
In this article we prove that the minimum-degree greedy algorithm, with adversarial tie-breaking, is a $(2/3)$-approximation for the \textsc{Maximum Independent Set} problem on interval graphs. We show that this is tight, even on unit interval graphs of maximum degree 3. We show that on chordal graphs, the greedy algorithm is a $(1/2)$-approximation and that this is again tight. These results contrast with the known (tight) approximation ratio of $\frac{3}{\Delta+2}$ of the greedy algorithm for general graphs of maximum degree $\Delta$.
\end{abstract}

\section{Introduction and Preliminaries}
\subsection{Background}
Let $G$ be an undirected graph. For a vertex $v \in G$ we write $N(v)$ to denote the set of neighbours of $v$, and $N[v]$ to denote the closed neighbourhood of $v$, i.e. $N(v) \cup \{v\}$. The degree of a vertex $v$ is denoted by $d(v)$. %= |N(v)|$.
We write $\Delta(G) = \Delta$ to denote $\max_{v \in G} d(v)$, and $\delta(G) = \delta$ to denote $\min_{v \in G} d(v)$. 
Set $I \subseteq V$ is an \emph{independent set} of graph $G$ if all the vertices in $I$ are pairwise non-adjacent in $G$. The computational problem \textsc{Maximum Independent Set (MIS)} asks to compute, for a given undirected graph $G$, an independent set of maximum size. The problem is NP-hard \cite{gareyjohnson}.
The \emph{minimum-degree greedy algorithm}, hereafter referred to as \textsc{Greedy}, computes an independent set $I$ of $G$ as follows: 
(1) Set $I := \emptyset$. 
(2) While $G$ still contains vertices, select a vertex $v$ where $d(v) = \delta(G)$, set $I := I \cup \{v\}$, set $G := G \setminus N[v]$, and continue.

This article concerns how close \textsc{Greedy} comes to computing a maximum independent set. To ensure that the question is well-defined, we need to clarify how \textsc{Greedy} breaks ties i.e. given multiple minimum-degree vertices, how it decides which one to pick. In this article we assume `adversarial' tie-breaking: that is, we study worst-case tie-breaking (contrasting with the variant when \textsc{Greedy} has `advice' available to help it break ties advantageously \cite{halldorsson1995greedy,krysta2020ultimate}).
This fits naturally with the worst-case analytical framework of \emph{approximation ratios} of algorithms \cite{papa91}. Here, an \emph{approximation ratio} of an algorithm is the infimum, ranging over all inputs $G$, of the ratio between the size of the solution returned by \textsc{Greedy} on $G$ and the size of a maximum independent set of $G$.
An algorithm for \textsc{Maximum Independent Set} that runs in polynomial time and for any input $G$ always computes an independent set $I_a$ such that $|I_a|\geq \rho\cdot |I^*|$, where $I^*$ is a maximum independent set of $G$, is called a \emph{$\rho$-approximation algorithm} (for \textsc{Maximum Independent Set}).
The value $\rho$ is called an \emph{approximation guarantee} of the algorithm.

Bodlaender et al.\, note that on trees, cycles, split graphs, complete $k$-partite graphs, and complements of $k$-trees \textsc{Greedy} is optimal \cite{bodlaender1997hard}. However, on general graphs, with adversarial tie-breaking, \textsc{Greedy} is a $\frac{3}{\Delta+2}$ approximation \cite{halldorsson1997greed},
and this approximation guarantee is is tight.
See also \cite{sakai2003note,kako2009approximation,boppana2020simple,Mestre06,korte1978analysis} for related literature.
Here we restrict our attention to \emph{interval} graphs and \emph{chordal} graphs. A graph $G$ is interval graph if the vertices can be modelled as closed intervals of the real numbers, such that two vertices are adjacent if the intersection of the two corresponding intervals is non-empty; such a family of intervals is called an \emph{interval representation} of $G$\footnote{\textsc{Maximum Independent Set} on interval graphs is also studied in the \emph{interval scheduling} literature \cite{kolen2007interval}, where the intervals are viewed as jobs with fixed start and end times.}. 
Interval graphs are a subclass of the chordal graphs. A graph is chordal if, for every simple cycle on 4 or more vertices, there exist two non-consecutive vertices on the cycle that are adjacent. Chordal graphs and interval graphs are closed under vertex deletion.

The problem \textsc{Maximum Independent Set} can be solved to optimality in linear time on chordal graphs \cite{rosetarjan}.
Probably because of this strongly positive result, the approximation ratio of \textsc{Greedy} has never been studied on interval or chordal graphs.
Here we close this gap in the literature by proving that on chordal graphs \textsc{Greedy} is a (1/2)-approximation, and that this approximation guarantee is tight. We then slightly strengthen this analysis to show that on interval graphs \textsc{Greedy} is a (2/3)-approximation, and that this approximation guarantee is tight, even on \emph{unit}-interval graphs of maximum degree 3. (A unit-interval graph is one having an interval representation in which all intervals have the same length). 
For completeness we then show that \textsc{Greedy} does not give constant-factor approximations on two natural generalizations of interval graphs and chordal graphs.
We conclude with a discussion of several open problems.

\subsection{Further preliminaries}
A \emph{clique} of a graph is a set of pairwise adjacent vertices.
A vertex $v$ is \emph{simplicial} if the graph induced by $N[v]$ is a clique. It is well known that every chordal graph $G$ has at least one simplicial vertex.
%, and at least two if $G$ has two or more vertices.
Note that if $d(v) \leq 1$ then $v$ is simplicial.

Clearly, if $G$ has more than one connected component, the size of a maximum independent set is equal to the sum of the maximum independent sets of the connected components. As various other authors have noted, the following also holds:
\begin{observation}
\label{obs:interleave}
The execution of \textsc{Greedy} on a disconnected graph $G$ can be analysed by studying the execution of \textsc{Greedy} independently on the connected components of $G$. 
\end{observation}
This is due to the fact that \textsc{Greedy} selecting a vertex in one connected component has no impact on vertices in the other connected components. Hence, the order that \textsc{Greedy} selects vertices within a single connected component is unaffected by its choices in other components.

The following two observations are also well known (see e.g. \cite{halldorsson1997greed}).

\begin{observation}
\label{obs:simplicial}
Let $G$ be an undirected graph and $v$ a simplicial vertex of $G$. There exists a maximum independent set that includes $v$.
\end{observation}

\begin{observation}
\label{obs:deg2}
Let $G$ be an undirected graph. If $\Delta \leq 2$ then \textsc{Greedy} is optimal.
\end{observation}
\begin{proof}
A graph with $\Delta \leq 2$ consists of isolated vertices, paths and cycles. \textsc{Greedy} is easily seen to be optimal on all three types of graph \cite{bodlaender1997hard}. Specifically, isolated vertices and paths always have a minimum degree of 0 and 1, and such vertices are simplicial, so Observation \ref{obs:simplicial} holds. For cycles, every vertex can be part of some maximum independent set, so whichever vertex \textsc{Greedy} chooses on the cycle will be optimal, and after this step, \textsc{Greedy} operates on a path, for which \textsc{Greedy} is optimal. 
\end{proof}

\section{Tight Bounds: Chordal and Interval Graphs}

In this section we give matching lower and upper bounds for the performance of \textsc{Greedy} on the classes of chordal and interval graphs. 

\subsection{Lower bounds on the approximation ratio}

Here we first prove that \textsc{Greedy} is a (1/2)-approximation algorithm on chordal graphs.
We then slightly strengthen the analysis to prove that \textsc{Greedy} is a (2/3)-approximation algorithm on interval graphs. 
First, we require a number of auxiliary results.

\begin{lemma}
\label{lem:notadjacent}
Let $G$ be a (not necessarily chordal) graph and let $v$ be the first vertex picked by \textsc{Greedy}. Suppose $v$ has $k \geq 2$ neighbours $u_1, \ldots, u_k$ that are pairwise non-adjacent, i.e. $\{u_1,\ldots,u_k\}$ is an independent set. Then each vertex $u_1, \ldots, u_k$ has at least $k-1$
neighbours that are not in $N[v]$.
%least one neighbour that
%is not a neighbour of $v$. 
\end{lemma}
\begin{proof}
Due to the fact that \textsc{Greedy} picked $v$ first, we have $d(u_i) \geq d(v)$ for each $u_i$. Suppose, for the sake of contradiction, that some vertex $u \in \{u_1, \ldots, u_k\}$ has at most $k-2$ neighbours outside $N[v]$. Observe that all neighbours of $u$ that lie inside $N[v]$, also lie in $N[v] \setminus \{u_1,\ldots,u_k\}$. Combined with the fact that $\{u_1, \ldots, u_k\} \subseteq N[v]$, we have $d(u) \leq (|N[v]| - k) + (k-2) \leq |N[v]|-2$. However, $d(v) = |N[v]|-1$, so $d(u) < d(v)$, contradiction.
%does not have any neighbours that are non-adjacent to $v$. Then $N(u) \setminus \{v\} \subseteq N(v) \setminus \{u\}$. Select any $u' \in \{u_1, \ldots, u_k\} \setminus \{u\}$; due to the fact that $k\geq 2$, $u'$ definitely exists. Clearly, $u' \not \in N(u)$, but $u' \in N(v)$, so $deg(v) \geq deg(u)+1$, contradiction. 
\end{proof}

A \emph{tree decomposition} of an undirected graph $G=(V,E)$ is a pair
$(\mathcal{B}, \mathbb{T})$ where $\mathcal{B} = \{B_1, \dots ,B_q\}$, $B_i \subseteq V(G)$, is a (multi)set of \emph{bags}
and $\mathbb{T}$ is a tree with $q$ nodes, which are in bijection with $\mathcal{B}$, and satisfy the following three properties:

\begin{enumerate}
    \item[(tw 1)] $\cup_{i=1}^q B_i = V(G)$;
   \item[(tw 2)] $\forall e = \{ u,v \} \in E(G), \exists B_i \in \mathcal{B} \mbox{ s.t. } \{u,v\} \subseteq B_i$;
    \item[(tw 3)] $\forall v \in V(G)$, all the bags $B_i$ that contain $v$ form a connected subtree of~$\mathbb{T}$. Let $\mathbb{T}_v$ denote this subtree of $\mathbb{T}$.
\end{enumerate}

%Given a tree decomposition $\mathbb{T}$ of a graph $G$, we denote by $V(\mathbb{T})$ the (multi)set of its bags and
%by $E(\mathbb{T})$ the set of its edges. Property (tw3) is also known as \emph{running intersection property}.

A clique $C$ of $G$ is \emph{maximal} if there does not exist a clique $C'$ such that $C \subsetneq C'$. It is well known that for a chordal graph $G$ there exists a tree decomposition
$(\mathcal{B}, \mathbb{T})$ of $G$ whereby $\mathcal{B}$ is exactly the set of maximal cliques in $G$ \cite{diestel2017}. If $G$ is an interval graph, we may furthermore assume that $\mathbb{T}$ is a path \cite{golumbic1985interval} -- we shall use this fact later.

Let $I$ be an arbitrary independent set of $G$, and let $v$ be the first vertex picked by \textsc{Greedy}. We say that this is a $j$-\emph{move with respect to $I$} if $|N[v] \cap I| = j$. Informally, if $I$ is a maximum independent set, and \textsc{Greedy} makes a 0-move or a 1-move, then it `keeps up with optimality'
but for $j \geq 2$, a $j$-move causes \textsc{Greedy} to fall behind optimality. However, as we shall see, $j$-moves $(j \geq 2)$ induce new 0-moves and 1-moves that will occur later in the execution of \textsc{Greedy}, such that aggregating over its entire execution \textsc{Greedy} never strays too far from an optimal solution.

%A crucial part of our argument will be to show that, aggregating over all iterations of \textsc{Greedy}, the number of 0-moves and 1-moves will be `large' compared to the number of $k$-moves ($k \geq 2$).

\begin{lemma}
\label{lem:disconnect}
Let $G$ be a chordal graph and let $v$ be the first vertex picked by $\textsc{Greedy}$. Let $I$ be an arbitrary independent set of $G$. If $v$ is a $j$-move with respect to $I$, $j \geq 2$, then $G \setminus N[v]$ has at least $(j-1)$ more connected components than $G$.
\end{lemma}
\begin{proof}
Without loss of generality we assume that $G$ is connected. 
Let $(\mathcal{B}, \mathbb{T})$ be a tree decomposition of $G$ where the bags are the maximal cliques of $G$. 
Let $I'=\{u_1, \ldots, u_j\}$ be the elements of $I$ that are adjacent to $v$ in $G$. 
Since each bag is a clique, no two elements of $I'$ can be in the same bag of the tree decomposition. 
Consider the subtrees $\mathbb{T}_v$, $\mathbb{T}_{u_i}$ as in (tw 3). 
Given that the edges $\{v,u_i\}$ exist for each $u_i \in I'$, it follows by (tw 2) that for each $u_i \in I'$, some bag contains both $v$ and $u_i$. Hence,
$\mathbb{T}_v$ must intersect with $\mathbb{T}_{u_i}$, for each $u_i \in I'$.
Observe that $N[v]$ is exactly equal to the union of all the bags that contain $v$ i.e. all the bags in the subtree $\mathbb{T}_v$. 
%Hence, the vertices that survive in $G \setminus N[v]$ are those vertices

\begin{figure}
\begin{center}
\includegraphics[scale=0.5]{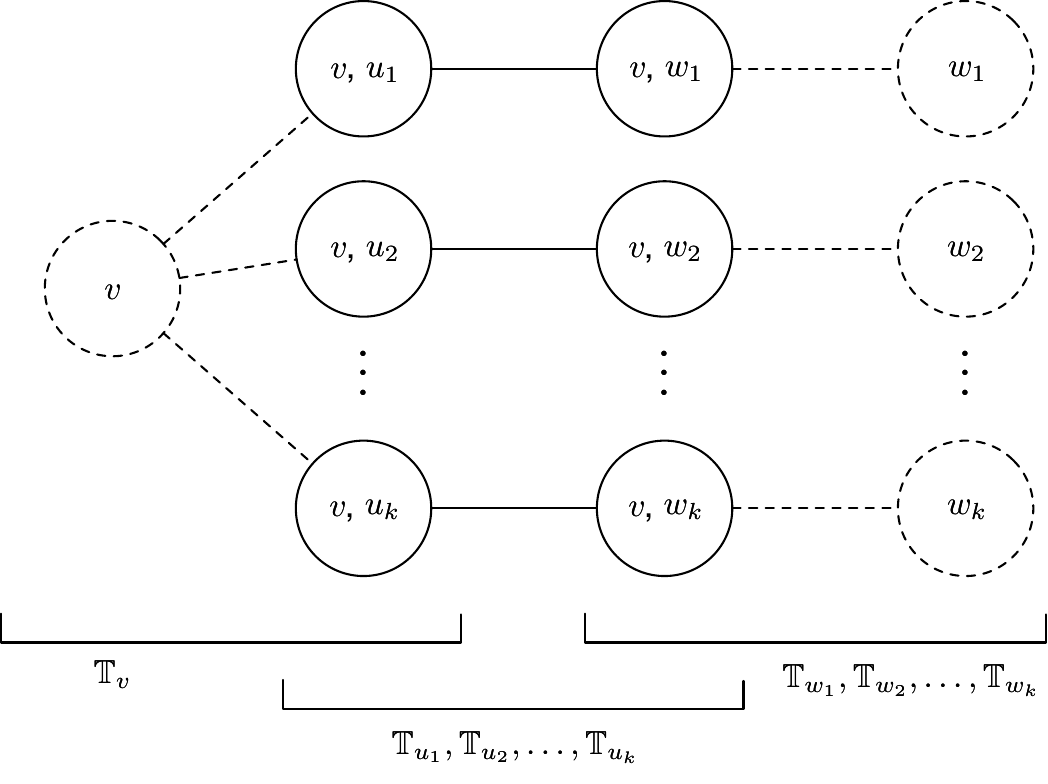}
\caption{A schematic depiction of the situation in the proof of Lemma \ref{lem:disconnect}. The vertices $w_1, \ldots, w_k$ will be in separate connected components of $G \setminus N[v]$. The set $N[v]$ is equal to the union of all the bags in $\mathbb{T}_v$.}
\label{fig:td}
\end{center}
\end{figure}

From Lemma \ref{lem:notadjacent} each node $u_i \in I'$ has some neighbour $w_i \neq v$ that is not adjacent to $v$. Hence, $w_i$ is not in any bag in the subtree $\mathbb{T}_v$ (so $\mathbb{T}_{w_i}$ is disjoint from $\mathbb{T}_{v}$). Next, we know that the edges $\{u_i, w_i\}$ exist,
so for each $u_i \in I'$ the subtrees $\mathbb{T}_{u_i}$ and $\mathbb{T}_{w_i}$ must intersect, due to (tw 2). 
Thus, for each $u_i \in I'$, $\mathbb{T}_{u_i}$ intersects with $\mathbb{T}_{v}$ but crucially also intersects with at least one bag (i.e. some bag of $\mathbb{T}_{w_i}$) that is \emph{not} part of $\mathbb{T}_{v}$. 
Now, let $X$ be the set of edges in $\mathbb{T}$ that have one endpoint in $\mathbb{T}_v$ and one point outside
it. Deleting the edges in $X$ naturally splits the tree decomposition into $\mathbb{T}_v$ and $|X|$ subtrees `pendant' to $\mathbb{T}_v$.

%we say that a subtree $\mathbb{T}'$ of the tree decomposition is a \emph{pendant subtree} with respect to $\mathbb{T}_v$ if it possible to delete a single edge of $\mathbb{T}$ to separate $\mathbb{T}_v$ from $\mathbb{T}'$.

The following situation is shown schematically in Figure \ref{fig:td}. 
We know that the $\mathbb{T}_{u_i}$ are mutually disjoint, so each of the  $\mathbb{T}_{w_i}$ is entirely contained within a distinct single subtree pendant to $\mathbb{T}_{v}$. (Each $\mathbb{T}_{w_i}$ can be contained in only one pendant subtree, because if it was in two or more such subtrees, by (tw 3) it would intersect with $\mathbb{T}_{v}$. 
Next: if, say, $\mathbb{T}_{w_1}$ and $\mathbb{T}_{w_2}$ would be contained within the same subtree pendant to $\mathbb{T}_{v}$, then $\mathbb{T}_{u_1}$ and $\mathbb{T}_{u_2}$ would necessarily  intersect each other, in the bag of $\mathbb{T}_v$ incident to the pendant subtree.) 
It follows from the properties of tree decompositions that the vertices $w_i$ are all in distinct connected components of $G \setminus N[v]$. Specifically, suppose (say) some path from $w_1$ to $w_2$ survives in $G \setminus N[v]$. Then some vertex of that path would necessarily have been in a bag in $\mathbb{T}_{v}$ (due to $w_1$ and $w_2$ being in distinct pendant subtrees). But such a vertex is in $N[v]$ and thus would have been deleted; so such a path cannot exist.

To conclude: $G \setminus N[v]$ contains at least $j$ connected components.
\end{proof}

Now, fix an arbitrary maximum independent set $I$ of $G$. 
At each iteration, we consider the vertex that \textsc{Greedy} selects \emph{with respect to the part of $I$ that is still available at the start of the iteration}. So, if $v$ is the first vertex chosen by \textsc{Greedy} then this is with respect to $I$, the second vertex chosen is with respect to $I \setminus N[v]$, and so on. This allows us to unambiguously extend the notion of a $j$-move to all iterations of the algorithm, not just the first. Let $m_j$, $j\geq 0$, denote the total number of $j$-moves made by \textsc{Greedy} during its execution. It follows that the size of the solution returned by \textsc{Greedy} is $\sum_{j \geq 0} m_j$, and
\begin{equation}
|I| = \sum_{j \geq 0} j \cdot m_j.
\end{equation}

\begin{lemma}
\label{lem:unavoidable}
Let $G$ be a (not necessarily connected) chordal graph and $I$ an arbitrary independent set of $I$. Let $k$ be the number of connected components in $G$. Then $m_0 + m_1 \geq k$.  
\end{lemma}
\begin{proof}
Connected components can be considered independently, due to Observation \ref{obs:interleave}. Hence, it is sufficient to prove the claim for $k=1$. We prove this by induction on the number of vertices in $G$. 
It can easily be verified that for every connected chordal graph $G$ with at most 3 vertices\footnote{These are: single vertex, 2-vertex path, 3-vertex path, triangle.}, and every not necessarily maximum independent set $I$ of $G$, the first step taken by \textsc{Greedy} is a 0-move or a 1-move with respect to $I$. 
Suppose, then, that $G$ has four or more vertices, and let $I$ be an arbitrary independent set of $G$. If the first step taken by greedy is a 0-move or a 1-move then we are done. If not, it is a $j$-move for some $j\geq 2$, which from Lemma \ref{lem:disconnect} disconnects $G$ into two or more connected components, all of which are smaller than $G$ and are chordal graphs. By induction each connected component has at least one 0-move or 1-move. The claim then follows (again, leveraging Observation \ref{obs:interleave}). 
\end{proof}

\begin{theorem}
\label{thm:chordal}
\textsc{Greedy} is at least a (1/2)-approximation on chordal graphs.    
\end{theorem}
\begin{proof}
The approximation ratio achieved by an execution of \textsc{Greedy} is
\begin{equation*}
\frac{\sum_{j \geq 0} m_j }
{\sum_{j \geq 0} j \cdot m_j}.
\end{equation*}
For convenience, we let $m_{0,1} = m_0 + m_1$. A lower bound on the
approximation ratio is thus
\begin{equation}
\label{eq:ratio}
\frac{m_{0,1} + \sum_{j \geq 2} m_j }
{m_{0,1} + \sum_{j \geq 2} j \cdot m_j} 
= 1 - \frac{\sum_{j\geq 2} (j-1) \cdot m_j}{m_{0,1} + \sum_{j \geq 2} j \cdot m_j}.  
\end{equation}

Now, suppose we prove the following:
\begin{equation}
\label{eq:bound}
m_{0,1} \geq 1 + \sum_{j \geq 2} (j-1) \cdot m_j.
\end{equation}
From (\ref{eq:bound}) it follows that a lower bound on (\ref{eq:ratio}) is
\begin{equation}
\label{eq:thirdbound}
1 - \frac{\sum_{j\geq 2} (j-1) \cdot m_j}{1 + \sum_{j \geq 2} (2j-1) \cdot m_j}.  
\end{equation}
Given that
$(2j-1) \geq 2(j-1)$ it follows that the absolute value of the fraction in (\ref{eq:thirdbound}) is strictly smaller than $(1/2)$. This shows that the approximation ratio remains strictly above (1/2), although to a vanishing degree.

It remains to prove (\ref{eq:bound}). To show that
$m_{0,1} \geq 1 + \sum_{j \geq 2} (j-1) \cdot m_j$, consider a  connected component $C$ of $G$. 
For convenience we refer to 0-moves and 1-moves collectively as $\{0,1\}$-moves. 
From Observation \ref{lem:unavoidable}, \textsc{Greedy} will require at least one $\{0,1\}$-move when processing $C$. Let us call this the \emph{original} $\{0,1\}$-move; this will account for the $+1$ in the final bound. Now, if \textsc{Greedy} makes a $j$-move when processing $C$, $j \geq 2$, then by Lemma \ref{lem:disconnect} this splits the component into $j$ or more components. Each of these components requires at least one $\{0,1\}$-move. We re-allocate the original $\{0,1\}$-move to one of the newly created components $C'$, and `assign' the  $j$-move to the $j-1$
$\{0,1\}$-moves in the (at least) $j-1$ other newly created components distinct from $C'$.
 By iterating this process, we injectively map each $j$-move to $j-1$ $\{0,1\}$-moves. By also taking the original $\{0,1\}$-move into account, which contributes a $+1$ term, we obtain the required bound. 
\end{proof}

\begin{theorem}
\label{thm:interval}
\textsc{Greedy} is at least a (2/3)-approximation on interval graphs.  
\end{theorem}
\begin{proof}
Recall the proof of Lemma \ref{lem:disconnect}. When $G$ is an interval graph, the tree $\mathbb{T}$ can assumed to be a path. Hence, $\mathbb{T}_v$ can have at most two pendant subtrees. Hence $|I'| \leq 2$. It follows that $m_j = 0$ for $j \geq 3$. Hence, the lower bound (\ref{eq:thirdbound}) immediately becomes
\begin{equation*}
\label{eq:fourthbound}
1 - \frac{m_2}{1 + 3m_2},  
\end{equation*}
so the approximation ratio is at least (2/3).
\end{proof}
\emph{Remark.} The proof of Theorem \ref{thm:interval} shows that the approximation ratio achieved
by \textsc{Greedy} on a chordal graph is influenced by the \emph{leafage} $\ell(G)$ of the graph $G$: the minimum number of leaves in a tree decomposition of $G$ whose bags are the maximal cliques of $G$ \cite{lin98}. Namely, $|I'| \leq \ell(G)$, and this inequality will hold at every iteration of the algorithm, since leafage is non-increasing under vertex deletion. Hence, $m_j = 0$ for $j \geq \ell(G)+1$, yielding a lower bound on the approximation ratio of,
\[
1 - \frac{\ell(G)-1}{2\ell(G)-1}.
\]

\subsection{Tightness}

Here we show that the results in the previous section cannot be improved.

\begin{figure}
\begin{center}
\includegraphics[scale=0.3]{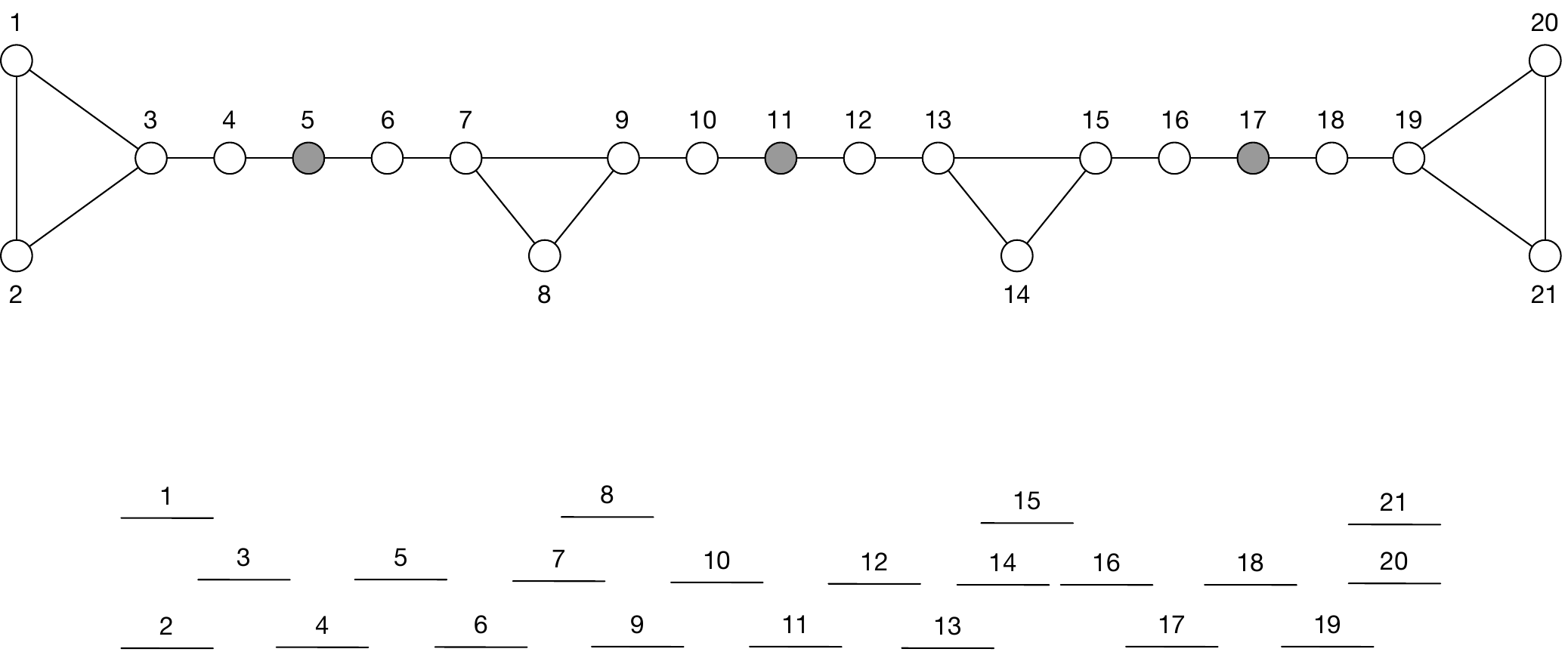}
\caption{The tight construction for interval graphs when $k=3$. In the top part of the figure is the graph itself, in the bottom part is a representation of the interval graph using unit-length intervals. The grey vertices indicate those which \textsc{Greedy} might pick first.}
\label{fig:lowerbound}
\end{center}
\end{figure}
\begin{lemma}
\label{lem:lowerbound}
 \textsc{Greedy} cannot be better than a (2/3)-approximation on interval graphs, even when restricted to unit interval graphs of maximum degree 3.
\end{lemma}
\begin{proof}
Consider the graph consisting of $k \geq 1$ paths and $k+1$ triangles as shown in Figure \ref{fig:lowerbound} (here $k=3$). This is a unit interval graph of maximum degree three. 
It is possible that \textsc{Greedy} starts by picking the $k$ grey degree-2 vertices. This yields a solution of size $k + (k+1) = 2k+1$. However, an optimum solution (obtained, for example, by eliminating simplicial vertices left to right) contains $2k + (k+1) = 3k+1$ vertices. This gives a  ratio of $\frac{2k+1}{3k+1}$ which is $(2/3) + o(1)$ as $k \rightarrow \infty$.
\end{proof}

\begin{lemma}
\label{lem:lowerboundChordal}
 \textsc{Greedy} cannot be better than a (1/2)-approximation on chordal graphs.
\end{lemma}
\begin{proof}
Fix $k \geq 2$. Create $k$ cliques $G_1,\ldots,G_k$ each of size $k+2$. 
In each $G_i$, delete an arbitrary edge $\{v_i, w_i\}$. 
Introduce a new vertex $u$ and connect it to each of $v_1, v_2, \ldots, v_k$. \textsc{Greedy} can pick $u$ and then one vertex from each of the $G_i \setminus \{v_i\}$ cliques, yielding a solution of size $k+1$. 
However, there exists an independent set of size $2k$, by picking all the $v_i, w_i$ vertices. 
This gives a ratio of $\frac{k+1}{2k}$ which is $(1/2) + o(1)$ as $k \rightarrow \infty$. 
\end{proof}

\section{Beyond Chordal and Interval Graphs: Perfect and 2-Track Interval Graphs}
\label{sec:perfect}

This section contains some obstacles to generalizing the results on chordal and interval graphs to obtain similar tight (constant) bounds on \textsc{Greedy} for related graph classes.
A well-known superset of chordal graphs are the \emph{perfect} graphs \cite{golumbic2004algorithmic}, so it is natural to ask whether \textsc{Greedy} yields a constant-factor approximation on perfect graphs. 
It does not. In fact, \textsc{Greedy} already fails on \emph{permutation} graphs (a subclass of perfect graphs incomparable to chordal graphs). A
graph $G$ on vertices $\{1, \ldots n\}$ is a permutation graph if there is a permutation $\pi$ of $\{1, \ldots n\}$ such that vertices $i<j$ are adjacent in $G$ if and only if $j$ occurs before $i$ in $\pi$ \cite[Chapter 7]{golumbic2004algorithmic}. 
%This is shown via the following construction. 

\begin{lemma}
\label{lem:permbad}
For every $k \geq 3$ there exists a permutation graph on $2k$ vertices with $\Delta=2k-2$, such that \textsc{Greedy} can return a solution of size 2 but a maximum independent set has size $k$.
\end{lemma}
\begin{proof}
Fix $k\geq 3$, and consider a permutation of $2k$ elements
\[
v_{k+1}, \ldots, v_{2k}, v_1, v_k, v_{k-1},\ldots, v_{2}.
\]
This induces a permutation graph (with respect to the identity permutation $v_1, \ldots, v_{2k}$) in which $\{v_{k+1}, \ldots, v_{2k}\} = I$ is an independent set,  $\{v_{2}, \ldots, v_{k}\} = C$ is a clique, $v_1$ is connected to everything in $I$, and every element in $I$ is connected to every element in $C$. Hence $v_1$ has degree $k$, as do the elements of $I$, and the elements of $C$ have degree $2k-2$. As $k\geq 3$, we have that $2k-2>k$. An optimal solution has size $k$, by picking $I$, but \textsc{Greedy} can pick $v_1$ and a single element from $C$, giving a solution of size 2.
\end{proof}

A natural generalization of interval graphs that do not yield a constant-factor approximation are the \emph{2-track interval} graphs \cite{fellows2009parameterized}. 
A 2-track interval graph is any graph that can be written as $G=(V, E_1 \cup E_2)$ where $(V, E_1)$ and $(V, E_2)$ are both interval graphs.

\begin{lemma}
\label{lem:2trackbad}
For every $k \geq 3$ there exists a 2-track interval graph on $2k$ vertices with $\Delta=2k-2$, such that \textsc{Greedy} can return a solution of size 2 but a maximum independent set has size $k$.
\end{lemma}
\begin{proof}
Fix $k \geq 3$. Let $(V,E_1)$ be an interval graph on $2k$ vertices $V$ consisting of a vertex $u_1$ that is connected to every vertex in an independent set $I=\{v_1, \ldots, v_k\}$, together with a collection of $k-1$ degree-0 nodes $\{u_2, \ldots u_k\} = X$. Let $(V, E_2)$ consist of the same vertices but now $u_1$ is a node of degree 0, $X$ is a clique, $I$ remains an independent set, and every vertex in $X$ is connected to every vertex in $I$. In $G=(V,E_1 \cup E_2)$ $u_1$ has degree $k$, as does each element of $I$, and this is the minimum degree since the vertices in $X$ each have degree $2k-2$. An optimal solution has size $k$, by picking $I$, but \textsc{Greedy} can pick $u_1$ and a single vertex from $X$, giving a solution of size 2.
\end{proof}

In both constructions the maximum degree $\Delta$ is $2k-2$, and the approximation ratio is $2/k$. Expressed in terms of $\Delta$ this shows that $\frac{4}{\Delta+2}$ is an upper bound on the approximation ratio of \textsc{Greedy} on permutation and 2-track interval graphs.

\section{Conclusion and Open Questions}

A number of natural open questions remain. First: for which other non-trivial graph classes does \textsc{Greedy} give constant-factor approximations?
We note in passing that in contrast to the negative results in Section \ref{sec:perfect}, for some graph classes a constant-factor approximation is trivial.   
For example, \emph{Planar} graphs always have a vertex of degree at most 5, so \textsc{Greedy} trivially selects at least $1/6$ of the vertices of the graph, and hence at least $1/6$ of any independent set. 
However, tight analysis would still be interesting in such cases. 

Next, \textsc{Greedy} has also been analysed when additional `advice' is available to helps the algorithm to break ties~\cite{halldorsson1995greedy,krysta2020ultimate}: this contrasts with the adversarial/no advice approach to tie-breaking assumed in our article. 
It would be interesting to explore \textsc{Greedy-with-advice} on interval and chordal graphs. 
For example, if $\Delta \leq 3$, \textsc{Greedy} on interval graphs can be made optimal by breaking ties in favour of simplicial vertices, which will exist at every iteration.
What kind of tie-breaking might be advantageous for $\Delta \geq 4$? Relatedly, we note the literature which studies the complexity of computing the best possible performance of \textsc{Greedy} on a given graph, or determining whether a certain constant-factor approximation is achieved on the graph \cite{bodlaender1997hard,krysta2020ultimate}. 
What does this complexity landscape look like for interval and chordal graphs?\\
\\
\textbf{Acknowledgements.} M. Frohn was supported by grant
OCENW.M.21.306 from the Dutch
Research Council (NWO).

\bibliography{main}{}
\bibliographystyle{plain}
\end{document}